\documentclass[onecolumn,floatfix,aps,nofootinbib]{revtex4}
\pdfoutput=1
\usepackage[dvips]{epsfig}
\usepackage[english]{babel}
\usepackage[all,cmtip]{xy}
\usepackage[utf8]{inputenc}
\usepackage{bbm}
\usepackage{verbatim}
\usepackage{array}
\usepackage{tikz}
\usepackage{amsmath}
\usepackage{bm} 
\usepackage{amsmath}
\usepackage{yfonts}
\usepackage{amsthm}
\usepackage{amsmath,amscd}
\usepackage{pst-plot}
\usepackage{slashed} 
\usepackage{array}
\usepackage{tikz-cd}
\usepackage{amsfonts}
\usepackage{graphicx}
\usepackage{amssymb}
 
\usepackage{hyperref}
\usepackage{cases}
\usepackage{indentfirst} 
\usepackage[titletoc]{appendix}
\usepackage{subeqnarray}
\usepackage{setspace}
\usepackage{indentfirst} 

\usepackage{gensymb}
\usepackage{xcolor}
\usepackage{calrsfs}

\usepackage{wasysym}
\usepackage[all,cmtip]{xy}

    \usepackage{amsmath,amsfonts,amssymb,amsthm,mathrsfs,bbm,braket}

    \newtheorem{prop}{Proposition}[section]

    \numberwithin{equation}{section}

\usepackage{bbold}

\begin{document}

\title{Breaking the permutation character of diffeomorphisms on spinor structures} 

\author{J. M. Hoff da Silva} 
\email{julio.hoff@unesp.br}
\affiliation{Departamento de F\'isica, Universidade Estadual Paulista, FEG-UNESP, Av. Dr. Ariberto Pereira da Cunha, 333, Guaratinguet\'a, SP, Brazil.}

\begin{abstract}
We investigate the impact of diffeomorphisms where more than one nonequivalent spinor structure is built upon a given base manifold endowed with nontrivial topology. We call attention to the fact that a relatively straightforward construction evinces a lack of symmetry between fermionic modes from different spinor bundle sections, leading to a dynamic preference breaking the permutation character of diffeomorphisms on spinor structures.   
\end{abstract}

\maketitle

\section{Introduction}

There is a precise sense under which spinors may be regarded as (kind of) scalars under diffeomorphisms. Directly stated, it needs a base manifold with trivial topology \cite{dab}. When nontrivial topology base manifolds are taken into account, then more than one nonequivalent spinor structure exist \cite{gre} and, since the nonequivalence forbids the existence of a global isomorphism between them, the (necessary and sufficient) condition for spinors behave as scalars under diffeomorphisms cannot be accomplished. Additionally, their action does not preserve spin structure, leading instead to a permutation of them. A permutation among spinor structures, as it is well-known, is a bijective map which, in turn, presupposes the absence of a preferred spinor structure. 

In the presence of nontrivial topology, the dynamics of some spinors are corrected by an additional term reflecting it \cite{sf,petry}. This extra term entering the Dirac operator shifts the dispersion relation so that the solution is not degenerated\footnote{Along this paper, degeneracy means different modes with same energy.} with the solution coming from the free (and without extra term) Dirac equation \cite{rev}. Thus, although a symmetry exists between solutions of different spinor bundle sections, the nondegeneracy between solutions indicates that diffeomorphisms act not so freely, engendering a permutation among spinor structures. In other words, the non-degenerate aspect of the spectra produces an asymmetry in the modes. The aforementioned asymmetry leads to the preferable spin structure because low-energy solutions are more likely to occur, and diffeomorphisms shall not jeopardize this aspect. Thus, the impact of diffeomorphisms on spinor structures differs from a usual permutation. This possibility raises questions about the theory of non-invariance under diffeomorphisms, but although further analysis is required, this property may be conceptually appropriate in some contexts involving diffeomorphism anomalies.   

In this paper, we evince the asymmetric aspect of solutions coming from different spinor structures through a relatively standard construction based on a simple base manifold exhibiting nontrivial topology and with two elements belonging to the first \v{C}ech cohomology group with coefficients in $\mathbb{Z}_2$. In Section II, we review some aspects presented in Ref. \cite{dab}, delineating the main ideas behind the lift of a given diffeomorphism to the spinor bundle structure, the reasoning entailing a usual permutation on spinor structures, and the consequent composition law. Going further, the possibility of a preferred spinor structure is constructed in section III, along with an analysis of possible composition laws. Section IV is reserved for final considerations.
    
\section{Diffeomorphisms impact on Spinor Structures}

Let $M$ be an $n-$dimensional orientable Riemannian manifold, and denote by $\text{Diff}(M)$ the group of diffeomorphisms of $M$ preserving orientability. Call $\text{Diff}_{\mathbb{1}}(M)\subset \text{Diff}(M)$ the subgroup of diffeomorphisms homotopic to the identity. Let $P_{GL^+(n)}$ denotes a principal bundle over $M$ with structure group given by $GL^+(n)$ and $P_{SO(n)}\subset P_{GL^+(n)}$ a principal subbundle. If ${_2GL}^+(n)$ is the $GL^+(n)$ double cover, then quite naturally $P_{{_2GL}^+(n)}\supset P_{Spin(n)}$. Aimed with the existence of $\rho:\,_2GL^+(n)\rightarrow GL^+(n)$ such that $\ker(\rho)=\mathbb{Z}_2\in Z( _2GL^+(n))$ and a bundle map $\eta:P_{_2GL^+(n)}\rightarrow P_{GL^+(n)}$, the restriction of $\eta$ to the subbundles provide $\eta|_{P_{Spin(n)}}:P_{Spin(n)}\rightarrow P_{SO(n)}$ and the pair $(P_{Spin(n)},\eta|_{P_{Spin(n)}})$ furnishes a spinor structure over $M$.

Now, consider $f\in\text{Diff}(M)$, whose pullback impact on the metric is given by $g'=f^* g$ as usual. The lift of $f$ to the bundle structure is an automorphism $\mathfrak{F}\in \text{Aut}(GL^+(n))$ whose restriction to $SO(n)$ maps orthonormal frames of the $g'$ metric on their $g$ metric counterparts yielding $\mathfrak{F}:P_{SO(n)|_{g'}}\rightarrow P_{SO(n)|_{g}}$. In this context 
\begin{equation}
\eta|_{P_{Spin(n)}}:P_{Spin(n)|_{g,g'}}\rightarrow P_{SO(n)|_{g,g'}},
\end{equation} that is, there exists a $\eta-$mapping for each metric. Out from the triple $(\mathfrak{F},\eta|{g'},\eta|{g})$ it can be constructed the isomorphism $P_{Spin(n)|_{g'}}\cong P_{Spin(n)|_{g}}$ by
\begin{equation}
\tilde{\mathfrak{F}}=\eta^{-1}|_g\circ \mathfrak{F}\circ \eta|_{g'}. \label{est}
\end{equation} The existence (unique up to isomorphisms) of $\tilde{\mathfrak{F}}\in \text{Aut}(Spin(n))$ is highly relevant to the behavior of spinors under diffeomorphisms. The following results express this relevance.

\begin{prop}      
A spinor field $\psi(x)$ is a scalar under a diffeomorphism $f$ --- in the sense that $\psi'(x)=\psi(f(x))$ --- if, and only if, $P_{Spin(n)|_{g'}}\cong P_{Spin(n)|_{g}}$.   
\end{prop}
\begin{proof} 
	Consider $\{r\}$ a field of orthogonal frames whose typical element $r$ entails $r:U\subset M\rightarrow P_{Spin(n)}|_g$ and for $U'=f^{-1}(U)$ admit $r':U'\rightarrow P_{Spin(n)}|_{g'}$. Additionally, take $\xi:P_{Spin(n)}|_g\rightarrow \mathbb{C}^4$ and $\xi':P_{Spin(n)}|_{g'}\rightarrow \mathbb{C}^4$. Therefore, a section of the spinor bundle may be written as $\psi=\xi\circ r$.   
 	
	($\Rightarrow$) Define $\xi'=\xi\circ \tilde{\mathfrak{F}}$ and $r'=\tilde{\mathfrak{F}}^{-1}\circ r\circ f$. Then $\psi'=\xi'\circ r'=\xi\circ r\circ f=\psi(f)$.
	
	($\Leftarrow$) If $\psi'(x)=\psi(f(x))$, then $\xi'\circ r'-\xi\circ r\circ f=0$. For $\kappa:P_{Spin(n)}|_g\rightarrow P_{Spin(n)}|_g$ and $\Omega:P_{Spin(n)}|_{g'}\rightarrow P_{Spin(n)}|_g$ define a general $\xi'$ transformation by $\xi'=\xi\circ\kappa\circ\Omega$. Similarly, define $r'=\beta\circ\sigma\circ r\circ f$ with  $\sigma:P_{Spin(n)}|_g\rightarrow P_{Spin(n)}|_g$ and $\beta:P_{Spin(n)}|_g\rightarrow P_{Spin(n)}|_{g'}$. Therefore 
	\begin{equation}
	\xi\circ(\kappa\circ\Omega\circ \beta\circ \sigma-1)\circ r\circ f=0,
	\end{equation} leading to $\phi\circ \Omega=\beta^{-1}$ for $\phi=\sigma\circ \kappa$. Hence $\phi\circ\Omega:P_{Spin(n)}|_{g'}\rightarrow P_{Spin(n)}|_g$ and the eventual difference between $\Omega$ and $\tilde{\mathfrak{F}}$ is compensated by the adjustment $\phi:P_{Spin(n)}|_{g}\rightarrow P_{Spin(n)}|_g$. 
\end{proof}	

The above proposition is to be contrasted with a complement and a broader contextualization: the existence of lifting $\tilde{\mathfrak{F}}$ is completely conditioned to the triviality of the first \v{C}ech cohomology group $\check{H}^1(M,\mathbb{Z}_2)$. On the other hand, being the base manifold topologically nontrivial, i.e., $\pi_1(M)\neq 0$, $\check{H}^1(M,\mathbb{Z}_2)$ is also nontrivial. Therefore, in nontrivial topology, spinors are not scalars under diffeomorphisms. By its turn, in the case of nontrivial $\check{H}^1(M,\mathbb{Z}_2)$, there is more than one nonequivalent spinor structure (actually as many spinor structures as elements in $\check{H}^1(M,\mathbb{Z}_2)$). Stated backward, when more than one spinor structure is in order, spinors do not behave as scalars under diffeomorphisms. The action of $f\in \text{Diff}(M)\backslash \text{Diff}_{\mathbb{1}}(M)$ at the level of the spinor structures set is to permute between them, or, in equivalent words, to permute among elements of $\check{H}^1(M,\mathbb{Z}_2)$. This permutation can indeed be constructed with the nontrivial elements of the first \v{C}ech cohomology group \cite{dab}. For instance, assume the existence of only two nonequivalent spinor structures. The action of diffeomorphisms upon them may be described as follows. Take $P\times\tilde{P}=\{(a,b)|a\in P_{Spin(n)}\times_{\nu} \mathbb{C}^4, b\in\tilde{P}_{Spin(n)}\times_{\nu} \mathbb{C}^4\}$, where $\nu$ stands for a spin $1/2$ Spin group representation. Let $f_1$ and $f_2$ belong to $\text{Diff}(M)\backslash \text{Diff}_{\mathbb{1}}(M)$. The sequential permutation lifting possibilities in $P\times\tilde{P}$ can be straightforwardly depicted denoting by $(a,b)$ the standard (S) situation and by $(b,a)$ the change (C) of spinor structures. Thus, we have 
\begin{eqnarray}
(a,b)(a,b)=(a,b),\nonumber\\ 
(a,b)(b,a)=(b,a),\nonumber\\
(b,a)(a,b)=(b,a),\nonumber\\
(b,a)(b,a)=(a,b),
\end{eqnarray}that is to say, naively, $(S)(S)=(S)$, $(S)(C)=(C)$, $(C)(S)=(C)$ and $(C)(C)=(S)$, rendering the group isomorphic to $\mathbb{Z}_2$ with $(S)=(a,b)$ playing the role of identity. The permutation aspect was further investigated in many contexts. For instance, the existence of automorphisms leaving one (or every) spin structure invariant was analyzed in Ref. \cite{kallel}, while the study of spin structures permutation via the Artin braid group was pursued in Ref. \cite{wang}. The bijection character of diffeomorphisms on spinor structures indeed allows for these generalizations. In the next section, we argue that this bijection can be broken in some cases, and a dynamic selection of preferred spin structures can be ordered.  

\section{A preferred spinor structure possibility}

Consider $\pi_1(M)\neq \{0\}$ so that there is more than one nonequivalent spinor structure. For the sake of the argument, consider the existence of only two spinor structures, say $P$ and $\tilde{P}$ for short. We start this section with the key observation that while the spinor structures are nonequivalent, there is a linear and invertible map at the level of the section of spinor bundles such that 
\begin{eqnarray}
\rho:&&\left.\sec \tilde{P}_{Spin(n)}\times_{\nu}\mathbb{C}^4\rightarrow \sec P_{Spin(n)}\times_{\nu}\mathbb{C}^4\right.\nonumber \\&&\hspace{2.5cm}\left. \tilde{\psi}\mapsto \psi=\rho\tilde{\psi}.\right.
\end{eqnarray} All the details for obtaining such a map can be found elsewhere \cite{petry,rev}. Here, we shall depict the main steps related to the desired construction. The map above relates exotic (tilde) spinors with usual ones. However, a simple derivative of a spinor is also a spinor and, as such, one must find out a derivative operator, say $\nabla$, such that $\partial\psi=\rho\nabla\tilde{\psi}$. Let us make this point more precise: define a topological compensating field $B:\sec \tilde{P}_{Spin(n)}\rightarrow \sec \tilde{P}_{Spin(n)}$, so that $\nabla=\partial+B$. It can be readily obtained $B$ (details can be found in Ref. \cite{rev}) defining a closed (but not exact) form with $\oint_{S^1}B=2\pi$, and providing the right transformation of $\nabla\tilde{\psi}$. 

The situation presents an interesting symmetry elicited by the following construction. Starting from $\rho(\partial\tilde{\psi}+B\tilde{\psi})=\partial\psi$, the Dirac operator can be related in both sections by $\rho(D^+_E\tilde{\psi})=D_0\psi$, where $D_0$ denotes the standard Dirac operator and $D^+_E$ stands for the exotic one. The `$+$' label explicitly states that the $B$ form is added to $D_0$. Consider the set $\Xi:=\{\psi\}\in \sec P_{Spin(n)}\times_{\nu}\mathbb{C}^4$ and the splitting of $\Xi$ into $\{\ker{D_0},\Xi \,\backslash \ker{D_0}\}$, i.e., standard Dirac and amorphous\footnote{Though the term ``amorphous'' is used in different contexts \cite{crum,wal} we shall keep this designation here for spinors not obeying the standard Dirac equation.} spinors, respectively. Consider also the exotic counterpart, $\tilde{\Xi}:=\{\tilde{\psi}\}\in \sec \tilde{P}_{Spin(n)}\times_{\nu}\mathbb{C}^4$. Notice that for spinors $\psi \in \ker{D_0}\subset \Xi$ we have $\rho(D^+_E\tilde{\psi})=0$ and since $\rho$ is invertible it implies the existence of $\tilde{\psi}\in \tilde{\Xi} \,\backslash \ker{D_0}$ such that $D^+_E\tilde{\psi}=0$. Now, take a natural injection $\iota: \sec \tilde{P}_{Spin(n)}\times_{\nu}\mathbb{C}^4 \hookrightarrow \sec P_{Spin(n)}\times_{\nu}\mathbb{C}^4$ such that 
\begin{eqnarray}
\rho(B\tilde{\psi})=\iota(B)\rho\tilde{\psi}=B\psi \in \sec P_{Spin(n)}\times_{\nu}\mathbb{C}^4. 
\end{eqnarray} Therefore, from $\rho(D_E\tilde{\psi})=D_0\psi$ we have 
\begin{eqnarray}
D_0\tilde{\psi}=\rho^{-1}(D_E^-\psi)
\end{eqnarray} and for the exotic spinors set $\{\tilde{\psi}\}\in \ker{D_0}$ there are corresponding usual spinors $\{\psi\in \Xi\,\backslash \ker{D_0}|D_E^-\psi=0\}$. 

The solutions for $D_0$ obviously degenerate and, since the bilinear covariants for both sections are identical\footnote{The $\rho$ map is performed by an operator endowed with a complex unimodular factor.}, there is no measurable distinction between them. The situation changes for solutions of $D_E^{\pm}$. These modes are not degenerated when compared to solutions of $D_0$ and, more importantly, not degenerated when compared between themselves. For convenience\footnote{The asymmetry we are about to investigate is also expected to happen in curved spaces. We should keep our discussion as simple as possible, as it suffices to illustrate the relevant results. The analysis can be further adapted to more general base (suitable) manifolds without much difficulty.}, assume the simplest case in which $M$ is the Minkowski space and consider a subspace $\Sigma\subset M$ upon which acts a diffeomorphism $f:\Sigma\rightarrow (\mathbb{R}^2\times S^1)\times \mathbb{R}$, with $f$ defined as before. Clearly $\pi_1(\Sigma)\neq \{0\}$ and $\check{H}^1(\Sigma,\mathbb{Z}_2)=\check{H}^1(S^1,\mathbb{Z}_2)\cong \mathbb{Z}_2$ due to $f$. Therefore, there are two (and only two) spinor structures.   

Calling $B=\gamma k$ (where $\gamma$ stands for gamma matrices), the four-vector $k$ encompasses the correction from nontrivial topology to the derivative operator. Pursuing the simplest case in which $k$ slowly varies with a very low magnitude vector in the semi-classical framework suffices. This situation allows for disregarding $k^2$ and $\partial k$ terms. The dispersion relation associated with spinors belonging to $\ker D_E^\pm$ can be readily computed within this approximation. Denoting $E^\pm$ the relativistic energy of a spin $1/2$ fermion of rest mass $m$ and momentum $\bf{p}$, the dispersion relation leads to  
\begin{equation}
E^\pm\approx (|{\bf p}|^2+m^2)\bigg[1\mp\frac{{\bf k}\cdot {\bf p}}{2(|{\bf p}|^2+m^2)}\bigg], 
\end{equation} where $k_0$ was settled to zero ensuring $m$ as the rest energy. The net topological effect is certainly very low in this approximation, but it is relevant that there is an energy dependence on the sign even in this case. This dependence breaks the degeneracy of solutions. To the argument, one could fix the momentum such that ${\bf k}\cdot {\bf p}>0$, and thus $\tilde{\psi}\in \ker D_E^+$ would be dynamically preferable. Of course, if things are such that ${\bf k}\cdot {\bf p}<0$, then the dynamic preference would fall on $\psi_E^-$. For the argument continuity, it suffices that one solution is preferable. Hence, for convenience, let us keep the analysis for ${\bf k}\cdot {\bf p}>0$. The central point is to understand how diffeomorphisms impact this degeneracy lifting. 

Let us denote, from now on, by $\tilde{\psi}_E^+$ an element of $\ker D_E^+$, by $\tilde{\psi}_0$ an element of $\ker D_0$ in $\sec \tilde{P}$ (for short), and analogously by $\psi_0 \in \sec P$ an element of $\ker D_0$ and by $\psi_E^-$ an element of $\ker D_E^-$. As shown, $\tilde{\psi}_E^+ \in \sec \tilde{P}$ and $\psi_0 \in \sec P$ are related by $\rho$, while the remain pair is related by $\rho^{-1}$. Moreover, since the ${\bf k}$ encodes a nontrivial topology correction to the dynamical operator, it is unlike that the solutions $\tilde{\psi}_E^+$ and $\psi_E^-$ can change from one spinor structure to another under diffeomorphism, since these modes are somewhat topologically protected (the nonholonomic condition $\oint_{S^1}B=2\pi$ cannot be undone without changing the base manifold topology). Besides, since the $\rho$ mapping connect $\tilde{\psi}_E^+$ and $\psi_0$, the robustness of $\tilde{\psi}_E^+$ under diffeomorphisms implies the same for $\psi_0$. One could think that the situation is the same for the remaining pair. However, the lack of degeneracy is such that a diffeomorphism changing the spinor structure of $\tilde{\psi}_0$ would imply a connection with the lowest energy mode again via $\rho$. The diagram below summarizes the situation, where the arrows denote spinor structure preferences under diffeomorphisms. 

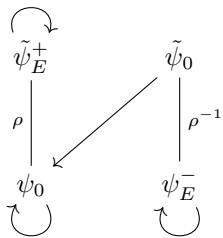
\[
\xymatrixcolsep{3pc}
\xymatrixrowsep{2.5pc}
\xymatrix{
	\tilde{\psi}_E^+ \ar@{-}[d]_-{\rho} 
	\ar@(ul,ur)[] 
	& \tilde{\psi}_0 \ar@{-}[d]^-{\rho^{-1}} \ar[dl] \\
	\psi_0 \ar@(dr,dl)[]  & \psi_E^- 
	\ar@(dr,dl)[] 
} 
\] \\

It is certainly the case of some additional prudence in this line of argumentation. One should be careful inferring the impact of diffeomorphisms on spinor structures from the behavior of sections (spinors themselves). However, we stress that in addition to the topological protection of modes $\tilde{\psi}_E^+$ there is a break of degeneracy, and a diffeomorphism must not destroy this asymmetry. Observe that the case ${\bf k}\cdot{\bf p}<0$ retails asymmetry but modifies the preferred mode and, consequently, the preferred spin structure. Hence, the diffeomorphism acting upon a given section lifts to action over spinor structures which, given the section dynamical preference, shall reflect such an asymmetry, preventing the permutation character of diffeomorphisms at the $P\times\tilde{P}$ space level.    

We propose to spam the $P\times\tilde{P}$ space into $\{(a,b), (ab,.), (.,ab),(b,a)\}$, where beyond the usual $(S)$ and $(C)$ possibilities, terms with the dot denote the preference to $P$ $(ab,.)$ or $\tilde{P}$ $(.,ab)$ possibilities under diffeomorphisms, thus breaking the bijection character necessary to permute among spinor structures. Within the context of our previous analysis, $P$ must present some degree of preference over $\tilde{P}$. A composition law taking it into account may be settled, for instance, allowing for the term $(ab,.)$ (and avoiding $(.,ab)$) whenever possible. This reasoning leads to the Cayley table below.      

\[
\begin{array}{c|cccc}
\text{$P\leftarrow \tilde{P}$} & (a,b) & (ab,.) & (.,ab) & (b,a) \\
\hline
(a,b) & (a,b) & (ab,.) & (a,b) & (b,a) \\
(ab,.) & (ab,.) & (ab,.) & (ab,.) & (ab,.) \\
(.,ab) & (a,b) & (ab,.) & (.,ab) & (b,a) \\
(b,a) & (ab,.) & (ab,.) & (b,a) & (a,b) \\
\end{array}
\]\\

Observe that the element $(.,ab)$ serves as the identity, the composition is abelian, and the preference makes $(ab,.)$ act as an attractor. Finally, changing the dynamical preference changes the Cayley table to the one below. Notice the abelian character permanence while changing the roles of identity and attractor terms, as expected.  

\[
\begin{array}{c|cccc}
\text{$P\rightarrow \tilde{P}$} & (a,b) & (ab,.) & (.,ab) & (b,a) \\
\hline
(a,b) & (a,b) & (a,b) & (.,ab) & (.,ab) \\
(ab,.) & (a,b) & (ab,.) & (.,ab) & (b,a) \\
(.,ab) & (.,ab) & (.,ab) & (.,ab) & (.,ab) \\
(b,a) & (.,ab) & (b,a) & (.,ab) & (a,b) \\
\end{array}
\]

Observe that, usually, the change ${\bf k}\mapsto -{\bf k}$ would require a diffeomorphism breaking (at least locally) orientability, a situation excluded under our hypothesis. Nevertheless, the vector ${\bf k}$ is taken as the gradient of a function related to the nontrivial topology \cite{rev}, i.e. ${\bf k}=\nabla\theta(x)$. This function may change under diffeomorphisms; hence, it could be the case of changing the sign of ${\bf k}$ (and consequently the preferred spin structure) while preserving orientation. This would require an additional involution $\theta(x)\mapsto -\theta(x)\, (\!\!\!\!\mod 2\pi)$ taking advantage of the $S^1$. This possibility adds difficulty in the diffeomorphism composition at $P\times \tilde{P}$ since it would require changing the above Cayley tables during successive compositions. Finally, it is clear that in the quite particular case of being the fermion momentum perpendicular to ${\bf k}$, the situation is completely degenerate, and the usual permutation applies.  

In an attempt to frame the non-invariance under diffeomorphisms into a consistent physical scenario, we want to finalize the discussion by pointing out that the indicated non-invariance under diffeomorphisms may not be a disadvantage. The properties of spinor structures transformation are known to be relevant in studying the (non)invariance of the effective action under diffeomorphisms \cite{isav}. Although a complete analysis must be performed, it is essential to evaluate the broken permutation proposed here as a type of `gauge fixing' procedure, with the potential to circumvent or at least mitigate global diffeomorphism anomalies. 

\section{Final Remarks}

Formally, the asymmetry investigated can be generalized to as many spinor structures as elements of $\check{H}^1(M,\mathbb{Z}_2)$; nevertheless, it would require extra attention to the degeneracy spectra and its lacking. At the moment, it is unclear whether the resulting analysis would necessarily be non-degenerated. In the case of partial degeneracy, it is intriguing to contrast the possibility raised here with the exhaustive approach presented in Ref. \cite{kallel}. 

The founded results suggest that the very concept of diffeomorphisms here shall not be understood in the usual sense but instead as a map between different theory phases. At the present stage of development, one should be careful in considering the diffeomorphisms as connecting nonequivalent vacuum since it would be too far in interpreting higher-energy field theoretical phenomena via a single-particle wave function. The use of non-unitary Bogoliubov transformations is a promising approach, but at this point, it is something we can only speculate on.

\section*{Acknowledgments}

The author thanks CNPq (grant No. 307641/2022-8) for financial support. It is a pleasure to thank Prof. Roldao da Rocha and Gabriel M. Caires da Rocha for useful conversations.


\begin{thebibliography}{10}	

\bibitem{dab} L. Dabrowski and R. Percacci, Spinors and diffeomorphisms, Commun. Math. Phys. {\bf 106}, 691 (1986).

\bibitem{gre} W. Greub and H. R. Petry, On the lifting of structure groups, in Differential Geometrical Methods in Mathematical Physics, II, Lecture Notes in Mathematics, Springer, New York (1978). 

\bibitem{sf} C. J. Isham, Spinor fields in four-dimensional space-time, Proc. Roy. Soc. Lond. A {\bf 364}, 591 (1978).

\bibitem{petry} H. R. Petry, Exotic spinors in superconductivity, J. Math. Phys. {\bf 20}, 231 (1979).

\bibitem{rev} J. M. Hoff da Silva, Foundational aspects of spinor structures and exotic spinors, [arXiv:2502.15471 [math-ph]] (2025).

\bibitem{kallel} S. Kallel and D. Sjerve, Invariant spin structures on Riemann surfaces, Annales de la Facult\'e des Sciences de Toulouse {\bf XIX}, 457 (2010).

\bibitem{wang} G. Wang, The artin braid group actions on the set of spin structures on a surface, Hokkaido Mathematical Journal {\bf 51}, 427 (2022). 

\bibitem{crum} A. Crumeyrolle, {\it Orthogonal and Sympletic Clifford Algebras}, Kluwer Acad.
Publ., Dordrecht (1991).

\bibitem{wal} W. A. Rodrigues, Q. A. G. de Souza, J. Vaz Jr., and P. Lounesto, Dirac-Hestenes Spinor Fields on
Riemann-Cartan Manifolds, Int. J. Theor. Phys. {\bf 35}, 1849 (1996).

\bibitem{isav} S. J. Avis and C. J. Isham, Lorentz gauge invariant vacuum functionals for quantized spinor fields in non-simply connected space-times, Nucl. Phys. B {\bf 156}, 441 (1979).


\end{thebibliography}
\end{document}